\documentclass{eptcs}
\usepackage[T1]{fontenc}
\usepackage[utf8]{inputenc}
\usepackage{amsmath,amsfonts,amssymb}
\usepackage{stmaryrd} % for semantic brackets [[ ]]
\usepackage{xspace}
\usepackage{color}
\usepackage{listings}
\usepackage{newlfont}
\usepackage{amsthm}

\usepackage{program}
\usepackage{multicol}
\usepackage{graphicx}
\usepackage{comment}
\usepackage{url}
\usepackage[noend]{algorithmic}
\usepackage{algorithm}
\usepackage{hyperref}
\usepackage{float}

\newtheorem{thm}{Theorem}

\newtheorem{defi}{Definition}

\usepackage{lipsum}
\newcounter{examplecounter}

\definecolor{lg}{rgb}{0.9, 0.9, 0.9}

\usepackage{color,listings}

\lstdefinelanguage{lustre}
    { morekeywords={
        imported, node, 
        bool, int, float,
        initial, let, tel, until, unless, type, var, when, whennot,
        match, if, then, else, state, do, done, resume, restart, returns, merge,
        pre, current, last, map, red, fill, default,
        fby, automaton, tail, implies, pre, assert, rule, query, var, declare, rel,
        PROPERTY, or, and, requires, ensures, observer, rule, query, Bool, Int, ite, not},
      morecomment=[l][\color{blue}]{--},
      morecomment=[s][\color{blue}]{(*}{*)},
      morecomment=[l][\color{blue}]{assert}
    }[keywords,comments]

\lstnewenvironment{lustre}
  {\lstset{language=lustre}}{}

\newcommand{\In}{\mathcal{I}}
\newcommand{\Ou}{\mathcal{O}}
\newcommand{\Lo}{\mathcal{L}}
\newcommand{\lett}{\mathbf{let}}
\newcommand{\inn}{\mathbf{in}}

\NumberProgramstrue

\definecolor{listingComment}{rgb}{0.3,0.7,0.3}
\definecolor{listingString}{rgb}{0.3,0.3,0.7}
\definecolor{listingBackground}{rgb}{0.95,0.95,0.95}

\providecommand{\event}{HCVS 2014}

\begin{document}

\title{Synthesizing Modular Invariants for Synchronous Code%
  \thanks{\tiny This work was partially supported by the ANR-INSE-2012
    CAFEIN project, and also by NASA Contract No. NNX14AI09G.
Copyright 2014 Carnegie Mellon University.
This material is based upon work funded and supported by the
Department of Defense under Contract No. FA8721-05-C-0003 with
Carnegie Mellon University for the operation of the Software
Engineering Institute, a federally funded research and development
center.
Any opinions, findings and conclusions or recommendations expressed in
this material are those of the author(s) and do not necessarily
reflect the views of the United States Department of Defense.
NO WARRANTY. THIS CARNEGIE MELLON UNIVERSITY AND SOFTWARE ENGINEERING
INSTITUTE MATERIAL IS FURNISHED ON AN “AS-IS” BASIS. CARNEGIE MELLON
UNIVERSITY MAKES NO WARRANTIES OF ANY KIND, EITHER EXPRESSED OR
IMPLIED, AS TO ANY MATTER INCLUDING, BUT NOT LIMITED TO, WARRANTY OF
FITNESS FOR PURPOSE OR MERCHANTABILITY, EXCLUSIVITY, OR RESULTS
OBTAINED FROM USE OF THE MATERIAL. CARNEGIE MELLON UNIVERSITY DOES NOT
MAKE ANY WARRANTY OF ANY KIND WITH RESPECT TO FREEDOM FROM PATENT,
TRADEMARK, OR COPYRIGHT INFRINGEMENT.
This material has been approved for public release and unlimited distribution.
DM-0001278.
}}

\author{Pierre-Lo\"ic Garoche
\institute{Onera, The French Aerospace Lab}
\and
Arie Gurfinkel
\institute{SEI / CMU}
\and
Temesghen Kahsai
\institute{NASA Ames / CMU}
}
\def\titlerunning{Synthesizing Modular Invariants for Synchronous Code}
\def\authorrunning{P-L. Garoche , A. Gurfinkel \& T. Kahsai}

\maketitle

\begin{abstract}
In this paper, we explore different techniques to synthesize modular invariants for synchronous code encoded as Horn clauses. Modular invariants are a set of formulas that characterizes the validity of predicates. They are very useful for different aspects of analysis, synthesis, testing and program transformation. 
We describe two techniques to generate modular invariants for code written in the synchronous dataflow language Lustre. The first technique directly encodes the synchronous code in a modular fashion. While in the second technique, we synthesize modular invariants starting from a monolithic invariant. Both techniques, take advantage of analysis techniques based on property-directed reachability. We also describe a technique to minimize the synthesized invariants.

\end{abstract}

\section{Introduction}\label{sec:intro}
%!TEX root = pap.tex

In this paper, we present an algorithm for synthesizing modular
invariants for synchronous programs. Modular invariants are useful for
different aspects of analysis, synthesis, testing and program
transformation. For instance, embedded systems often contains complex
modal behavior that describe how the system interacts with its
environment. Such modal behaviors are usually described via
hierarchical state machines (HSM). The latter are used in model-based
development notations such as Simulink and SCADE --- the de-facto
standard for software development in avionics and many other
industries. For the purpose of safety analysis, Simulink/SCADE models
are compiled to a lower level modeling language, usually a synchronous
dataflow language such as Lustre~\cite{DBLP:conf/popl/CaspiPHP87}. Preserving the original
(hierarchical and modular) structure of the model is paramount to the
success of the analysis process. In this paper, we illustrate a
technique to preserve such structure via a modular compilation
process. Specifically, our techniques consists of compiling in a
modular fashion Lustre programs into Horn clauses.

The use of \emph{Horn clauses} as intermediate representation for
verification was proposed in~\cite{GuptaPR11}, with the verification
of concurrent programs as the main application. The underlying
procedure for solving sets of recursion-free Horn clauses, over the
combined theory of Linear Rational Arithmetic (LRA) and Uninterpreted
Functions (UF), was presented in~\cite{GuptaPR11b}. A range of further
applications of Horn clauses includes inter-procedural/exchange format
for verification problems that is supported by the SMT solver
Z3~\cite{pdr}. In this paper, we show how to use such techniques to
generate modular invariants for Lustre programs.

While on one hand we generate modular invariants by encoding the
synchronous code in a modular fashion, on the other hand, we are
interested to synthesize modular invariants from a monolithic
invariant. That is, given an invariant for a program that is obtained
by flattening the hierarchical structure, we want to reconstruct the
modular invariant. We describe a technique that generates modular
invariants from a monolithic one.

Finally, once we obtain modular invariants, we are interested in
minimizing such invariants. In particular, in our setting the
invariants correspond to a contract that a component must satisfy. In
general, the contract is state-full, i.e., it is an automaton (or a
protocol). In practice, it is important to generate contracts with
minimal state (i.e., minimal automata).  We sketch a direction for
minimization based on a CEGAR-like~\cite{cegar} technique. 

In summary, this paper makes the following contributions:
\begin{itemize}
\item Two techniques to generate modular invariants for synchronous
  code. One based on a modular compilation of Lustre code into Horn
  clauses, and a second one based on extracting modular invariants
  from monolithic invariants.
\item An implementation of these techniques that targets programs written in the synchronous dataflow language Lustre.
\item A sketch of a technique to minimize invariants using CEGAR-type
  approach.
\end{itemize}

The rest of the paper is organized as follows. In the next section, we
introduce synchronous languages in general and the synchronous
dataflow language Lustre. In Section~\ref{sec:modular}, we describe the first technique a
procedure to compile in a modular fashion Lustre code into Horn
clauses. In Section~\ref{sec:mono}, we illustrate the second technique to
derive modular invariants from a monolithic one. We conclude the paper
in Section~\ref{sec:min} with a discussion on the minimization of
invariants.
% \begin{enumerate}

% \item The preservation of the module at the verification stage allow us also to include auxiliary invariants to a specific module of the system without touching other parts of the system, i.e., \textit{targeted auxiliary invariants}.

% \item Generation of permissive interface, i.e., function contracts
% \begin{itemize}
% \item Modular proof allows to summarize required behaviour of a particular module.

% \end{itemize}

% \item Coverage analysis and vacuity detection
% \begin{itemize}
% \item Is the set of properties adequate to certify correctness of the model?

% \item Does the model satisfy the properties for the right reasons?
% \end{itemize}

% \item Evolving proofs / inductive invariants over program transformation.

% \end{enumerate}

%%% Local Variables: 
%%% mode: latex
%%% TeX-master: "pap"
%%% End: 

%  LocalWords:  invariants HSM Simulink SCADE de facto dataflow Lustre LRA UF
%  LocalWords:  SMT automata CEGAR

\section{Preliminaries}\label{sec:prelim}
%!TEX root = pap.tex
Synchronous languages are a class of languages proposed for the design of so
called ``reactive systems'' -- systems that maintain a permanent interaction
with physical environment. Such languages are based on the theory of synchronous
time, in which the system and its environment are considered to both view time
with some ``abstract'' universal clock. In order to simplify reasoning about
such systems, outputs are usually considered to be calculated
instantly~\cite{Benveniste91thesynchronous}. Examples of such languages include
Esterel~\cite{esterel}, Signal~\cite{signal} and
Lustre~\cite{DBLP:conf/popl/CaspiPHP87}. In this paper, we will concentrate on the
latter. Lustre combines each data stream with an associated clock as a mean to
discretize time. The overall system is considered to have a universal clock
that represents the smallest time span the system is able to distinguish, with
additional, coarser-grained, user-defined clocks. Therefore the overall system
may have different subsections that react to inputs at different frequencies. At
each clock tick, the system is considered to evaluate all streams, so all values
are considered stable for any actual time spent in the instant between ticks. A
stream position can be used to indicate a specific value of a stream in a given
instant, indexed by its clock tick. A stream at position $0$ is in its initial
configuration. Positions prior to this have no defined stream value. A Lustre
program defines a set of equations of the form:
$$
y_1, \dots, y_n = f(x_1,\dots, x_m, u_1,\dots, u_o)
$$
where $y_i$ are output or local variables and $u_i$ are input variables. Variables in Lustre are
used to represent individual streams and they are typed, with basic types
including streams of \textit{Real} numbers, \textit{Integers}, and
\textit{Booleans}. Lustre programs and subprograms are expressed in terms of
\textit{Nodes}. Nodes directly model subsystems in a modular fashion, with an
externally visible set of inputs and outputs. A \textit{node} can be seen as a
mapping of a finite set of input streams (in the form of a tuple) to a finite
set of output streams (also expressed as a tuple). The \textit{top node} is the main node of the program, the one that interface with the environment of the program and never be called by another node.

At each instant $t$, the node takes in the values of its input streams and returns the values of its output
streams. Operationally, a node has a cyclic behavior: at each cycle $t$, it
takes as input the value of each input stream at position or instant $t$, and
returns the value of each output stream at instant $t$. This computation is
assumed to be immediate in the computation model. Lustre nodes have a
limited form of memory in that, when computing the output values they can also
look at input and output values from previous instants, up to a finite limit
statically determined by the program itself. Figure~\ref{fig:lustre_synt}
describes a simple Lustre program: a node that every four computation steps
activates its output signal, starting at the third step. The \texttt{reset} input
reinitializes this counter.

%\vspace{-2em}
\begin{figure}[!h]
\centering
\begin{lustre}
                         node counter(reset: bool) returns (active: bool);
                           var a, b: bool;
                           let
                             a = false -> (not reset and not (pre b));
                             b = false -> (not reset and pre a);
                             active = a and b;
                           tel
\end{lustre}
%\vspace{-2em}
\caption{A simple Lustre example.}
\label{fig:lustre_synt}
\end{figure}
%\vspace{-1em}

Typically, the body of a Lustre node consists in a set of definitions, stream
equations of the form $x = t$ (as seen in Figure~\ref{fig:lustre_synt}) where
$x$ is a variable denoting an output or a locally defined stream and $t$ is an
expression, in a certain stream algebra, whose variables name input, 
output, or local streams. More generally, $x$ can be a tuple of stream variables 
and $t$ an expression evaluating to a tuple of the same type. Most of Lustre's operators are point-wise lifting to streams of the usual
operators over stream values. For example, let $x = [x_0, x_1, \dots]$ and $y =
[y_0, y_1, \dots]$ be two integer streams. Then, $x + y$ denotes the stream
$[x_0 + y_0; x_1 + y_1, \dots]$; an integer constant $c$, denotes the constant
integer stream $[c,c, \dots]$. Two important additional operators are a unary
shift-right operator \textit{pre} (``previous''), and a binary initialization
operator $\rightarrow$ (``followed by"). The first is defined as $pre(x) = [u,
x_0, x_1, \dots]$ with the value $u$ left unspecified. The second is defined as
$x \rightarrow y = [x_0, y_1, y_2, \dots]$. Syntactical
restrictions on the equations in a Lustre program guarantee that all its streams are
well defined: e.g. forbidding recursive definitions hence avoiding algebraic loops.

%%% Local Variables: 
%%% mode: latex
%%% TeX-master: "pap"
%%% End: 

%  LocalWords:  Esterel Lustre discretize Booleans tuple reinitializes bool pre
%  LocalWords:  unary

\section{Modular synthesis}\label{sec:modular}
%!TEX root = pap.tex

In the last section we gave an informal overview of the synchrounous dataflow language Lustre. A formal
semantics of Lustre is described in~\cite{DBLP:conf/popl/CaspiPHP87}. In this section, we
describe our technique to generate modular Horn clauses starting from Lustre code.

A Lustre program $L$ is a collection of nodes $[N_0, N_1, \dots, N_m]$
where $N_0$ is the top node, i.e., the main function. Each node is
represented by the following tuple:
$$
N_i = (\In_i, \Ou_i, \Lo_i, Init_i, Trans_i)
$$
where $\In_i, \Ou_i$ and $\Lo_i$ are set of input, output and local
variables. $Init_i$ and $Trans_i$ represents the set of formulas for the initial
states and the transition relation respectively, and they are defined as
follows:
\[
\bigwedge_{i \in \mathbb{N}} v_i = \rho (s_i)
\]
where
\begin{itemize}
\item $v_i \in \Ou_i \cup \Lo_i$ and $s_i$ is the expression such that
  $Vars(s_i) \subseteq \Ou_i \cup \In_i\cup \Lo_i$. $Vars(s_i)$ is the set of
  variables in $s_i$;
\item expressions $s_i$ are arbitrary Lustre expression including node calls
  $N_j(u_1,\ldots, u_n)$;
\item $\rho$ function maps expression to expression and projects the binary initialization operators
  $\rightarrow$:  \[
  a \rightarrow b\ \textrm{is projected as}\ \left\{
    \begin{array}{l}
      a\ \textrm{in}\ Init_i\\b\ \textrm{in}\ Trans_i
\end{array}
\right.
  \]
\end{itemize}
%
% \begin{enumerate}
% \item $Init = \bigwedge_{i \in \mathbb{N}} v_i = s_i $, $v_i \in \Ou_i
%   \cup \Lo_i$ and $s_i$ is the expression such that $Vars(s_i) \in \Ou_i \cup
%    \In_i\cup \Lo_i$. $Vars(s_i)$ is the set of variables in $s_i$.
% \item $Trans = \bigwedge_{i \in \mathbb{N}} v_i = s_i$ which have the same
%   condition as in $Init$, however, $s_i$ can contain also variables from node
%   calls. The latter are of the form $v_i = N_j(u_1,\ldots, u_n)$, where $N_j$ is
%   the name of the called node and $u_1,\ldots, u_n\in \In_j$ are the
%   variables of the called node. In case the expression is of the form $t \rightarrow e$,
%   then $s_i$ will only contain the sub-expression $e$.
% \end{enumerate}
%
Given a Lustre program $L=[N_0, N_1, \dots, N_m]$, a safety property $P$
is any expression over the signature over the main node $N_0$. A common
way to express safety property in synchronous languages is the use of
synchronous observers~\cite{Rushby:2012}. The
latter is a wrapper used to test observable properties of a node $N$ with
minimal modification the node itself; it returns an error signal if the property
does not hold, reducing the more complicated property to a single Boolean stream
where we need to check if the stream is constantly true.

We now describe the compiler $\mathtt{lus2horn}: L \rightarrow H$, which given a Lustre
program $L=[N_0, N_1, \dots, N_m]$ generates a set of Horn clauses $H$ that are semantically equivalent to $L$. The current compiler only handles a simplified version of the
Lustre v4 language without the constructs to manipulate clocks, or complex data
structure. The following steps describe the various stages of $\mathtt{lus2horn}$:

\textbf{Normalization:}
In the first phase the compiler $lus2horn$ transforms the equations of the
Lustre node to extract the stateful computations that appear inside
expressions. Stateful computation can either be the explicit use of a
\textit{pre} construct or the call to another node which may be stateful. The
extraction is made through a linear traversal of the node's equations,
introducing new equations for stateful computation\footnote{As opposed to 3
  addresses code, only the stateful part of expression is extracted.}. When possible, tuple
definition are split as simpler definitions. To ease later computation, each
node call is labeled by a unique identifier. The following set of expressions give an example of such normalization:

\[
\begin{array}{lcl}
  a\ =\ false \rightarrow (\mathbf{not}\ reset\ \mathbf{and\ not}\ (\mathbf{pre}\ b));
  &\longrightarrow& \left\{ \begin{array}{l}pb\ =\ \mathbf{pre}\ b;\\
a\ =\ false \rightarrow (\mathbf{not}\ reset\ \mathbf{and\ not}\ pb);
   \end{array} \right.\\\\

  y\ =\ 3 + node(x, 2);
  &\longrightarrow& \left\{ \begin{array}{l}\textrm{$res\_node1 = node^{uid_1}(x,
        2);$}\\ y\ =\ 3 + res\_node1;\end{array} \right.\\\\
\end{array}
\]
The following definitions are the normalization functions $Norm_{N}$, $Norm_{Eq}$ and $Norm_{Expr}$, a single node $N$, an equations and an expression respectively.
\begin{itemize}
\item The normalization of an expression returns a modified expression along
  with a set of newly bound stateful expressions and associated new variables:
$$
\begin{small}
\begin{array}{l}
  Norm_{Expr} (e, Eqs, Vars) \triangleq \\%\left\{
  \begin{array}{lcl}
    v & \rightarrow & v, Eqs, Vars\\
    cst &\rightarrow & cst, Eqs, Vars\\
   \mathbf{op}(e_1, \ldots , e_n) &\rightarrow &
    \left\{\begin{array}{l}
      \lett\ e_1, Eqs, Vars = Norm_{Expr}(e_1, Eqs, Vars)\ \inn\\
        \vdots \\
        \lett\ e_n, Eqs, Vars = Norm_{Expr}(e_n,
          Eqs, Vars)\ \inn\\
          \quad \quad op(e_1', \ldots, e_n'), Eqs, Vars
        \end{array}\right.
        \\
        \mathbf{pre}\ e &\rightarrow & \left\{
        \begin{array}{l}
          \lett\ e', Eqs, Vars = Norm_{Expr} (e, Eqs,
          Vars)\ \inn\\
          \lett\ x \notin Vars\ \inn \\
          \quad \quad x, \{x = pre e';\} \cup Eqs, \{x\} \cup Vars
        \end{array}\right.
        \\
    N_i(e_1, \ldots , e_n) &\rightarrow &
    \left\{\begin{array}{l}
      \lett\ e_1, Eqs, Vars = Norm_{Expr}(e_1, Eqs, Vars)\ \inn\\
        \vdots \\
        \lett\ e_n, Eqs, Vars = Norm_{Expr}(e_n,
          Eqs, Vars)\ \inn\\
          \lett\ x \notin Vars\ \inn \\
          \quad \quad x, \{x = N_i^{uid}(e_1', \ldots, e_n');\} \cup Eqs, \{x\} \cup Vars

        \end{array}\right.
      \end{array}
%\right.
\end{array}
\end{small}
$$

In $Norm_{expr}$, $op$ is a Lustre operator, $N_i$ is a node in a Lustre program $L$ and $uid$ is a unique
  identifier associated to the call of $N_i$ with arguments $(e_1', \ldots,
  e_n')$.

\item Normalization of a node equation simplifies tuple definitions and
  normalizes each expression. It returns a set of equations with normalized expressions:
\[
\begin{small}
\begin{array}{l}
Norm_{Eq} (\{v_1, \ldots, v_n = s\}, Eqs, Vars) \triangleq\\
\begin{array}{lcl}
v_1, \ldots, v_n = s_1, \ldots, s_n & \rightarrow &
\left\{\begin{array}{l}
\lett\ Eqs, Vars = Norm_{Eq} (\{v_1 = s_1\}, Eqs, Vars)\ \inn\\
\vdots\\
\quad \quad Norm_{Eq} (\{v_n = s_n\}, Eqs, Vars)
\end{array}\right.\\
v = s & \rightarrow &
\left\{\begin{array}{l}
\lett\ s', Eqs, Vars = Norm_{Expr} (e, Eqs, Vars)\ \inn\\
\quad \quad \{ v = s' \} \cup Eqs, Vars
\end{array}\right.
\end{array}
\end{array}
\end{small}
\]

\item Last the normalization of a node amounts to normalize each expression in
  each definition; the newly bound variables are added to the set of local variables.
\[
\begin{small}
\begin{array}{l}
Norm_N(N) = (\In_i,\ \Ou_i,\ \Lo_i \cup NewVars,\ {Init}_i,\ {Trans}_i)\\
\textrm{where}\ \left\{
\begin{array}{l}
\lett\ InitVars = \In_i \cup \Ou_i \cup \Lo_i\ \inn\\
\lett\ {Init_N}_i, Vars = Norm_{Eq} (Init_i, InitVars) \ \inn\\
\lett\ {Trans_N}_i, Vars = Norm_{Eq} (Trans_i, Vars) \ \inn\\
\quad \quad NewVars = Vars \setminus InitVars
\end{array}\right.
\end{array}
\end{small}
\]

%\arie{This is hard to follow. needs to be broken done and explained a bit more. }
% ploc: i detailed above. Hope this is better
\end{itemize}

  \textbf{State computation:} The state of a node is characterized by its
  memories: variables defined by \textit{pre} constructs, as well as the
  memories associated to each of its calling node instances.

We first define the set of local memories for a node:
\[
xMem (N_i) %\In_i,\ \Ou_i,\ \Lo_i,\ {Init_N}_i,\ {Trans_N}_i)
= \{ v \in \Lo_i \mid
\{v = pre\ e;\} \in Trans_i \}
\]

Then we characterize the set of callee instances, using their unique identifies $uid$:

%\arie{This needs to be defined recursively.}
\[
Inst (N_i) %\In_i,\ \Ou_i,\ \Lo_i,\ {Init_N}_i,\ {Trans_N}_i)
= \{ (N_j, uid)  \mid
\{v = N_j^{uid} (e_1, \ldots, e_n);\} \in {Trans_N}_i \}
\]

We denote by $State_i$ the set of memories fully characterizing the
state of a $N_i$ node instance\footnote{By construction, circular definition of nodes
are forbidden in Lustre: this recursive definition is well-founded.}.

%\arie{This is where introducing new state variables should be described.or in the 3-address-code transformation where pre-expressions are factored out. }

\[
State(N_i) = Mem(N_i) \cup \{ uid_{N_j}\_v \mid (N_j, uid) \in Inst(N_i) \land v \in
State(N_j) \}
\]

\textbf{Generating Horn predicates:} Once the memories of a node have been
identified, a predicate describing the transition relation can be expressed as a
Horn clause. The latter is defined over inputs, outputs, previous value of the
node's state $State_i$ and updated state $State'_i$. We then produce the
following Horn rule encoding the transition relation predicate:

\[
\begin{array}{lcr}
(i) & \quad\quad\quad & (\mathtt{rule}\ (\Rightarrow\ \phi(Trans_i)\ (T_{N_i}\ \In_i\ \Ou_i\ State_i\ State'_i)))
\end{array}
\]

% \ag{I don't like that we are mixing SMT-LIB syntax and logic. First, the word
% \texttt{rule} is just a short-hand for ``all free variables are
% universally quantified''. This is a standard assumption, so we can
% just state it and drop quantifiers. Second, we are mixing SMT-LIB and
% logic. For example, $\phi$ is not defined in SMT-LIB.}
Here, we use the Horn clause format introduced in Z3~\cite{pdr}, where
$(\mathtt{rule}\ expr)$ universally quantify the free variables of the SMT-LIB
expression $expr$.%declares a rule, where $rule$ is an SMT-LIB expression.
The function $\phi (expr)$ is recursively defined as

\[
\begin{small}  \begin{array}{lcl}
    \phi(v = pre\ e;) & \rightarrow & (= v'\ \ e)\\
    \phi(v_1, \ldots, v_n = N_j^{uid}(e_1, \ldots, e_m);) &\rightarrow &
    (T_{N_j}\ e_1\ \ldots\ e_m\ v_1\ \ldots\ v_n\ uid_{N_j}\_v_1\ \ldots\ uid_{N_j}\_v_k\
    uid_{N_j}\_v'_1\ \ldots\ uid_{N_j}\_v'_k)\\
   \phi (v = e;) &\rightarrow & (= v\ \ e)\\
   \phi( eq; eqs ) & \rightarrow & (\mathbf{and}\ (\phi(eq))\ (\phi (eqs))
\end{array}
\end{small}
\]
where
\begin{itemize}
\item $v$ in $v = pre\ e$ is by construction in $Mem(N_i) \subseteq State_i$ and
  $v' \in State'_i$;
\item $(uid_{N_j}\_v_l)_{1 \leq l \leq k} \in State_i$ denotes the state representation of
    the instance $uid$ of node $N_j$ and\\ $(uid_{N_j}\_v'_l)_{1 \leq l \leq k} \in State'_i$.
\end{itemize}

Similarly the Horn rule encoding the initial state is defined as follows:
\[
\begin{array}{lcr}
(ii) & \quad\quad\quad & (\mathtt{rule}\ (\Rightarrow\ \phi(Init_i)\ (I_{N_i}\ \In_i\ \Ou_i\ State'_i)))
\end{array}
\]

Given a Horn clause $H$ of the form $(\mathtt{rule}\ \Rightarrow\ Body\ B)$,
where $Body$ is a conjunction of expression and $B$ is a predicate, a model $\pi
: B \mapsto \mathcal{F}$ is a mapping from $B$ to a set of formulas
$\mathcal{F}$ such that it makes every rule $H$ valid. In other words, $\pi$
represents the set of invariants for the Horn clause $H$.

Let $(Main\ \In_{N_0}\ \Ou_{N_0}\ State_{N_0})$ be the predicate encoding the collecting semantics of
the main node $N_0$ of the Lustre program $L=[N_0, N_1, \dots, N_m]$. Each node
$N_i$ being defined by the two Horn clauses $I_{N_i}$ and $T_{N_i}$. The
semantics of the whole Lustre program is inductively encoded as follows:

$$
\begin{array}{l}
% (i) & (\mathtt{rule}\ \Rightarrow  & (\mathbf{and}\ \phi(Init)) (Top\_Init\ \In\ \Ou\ State\ State'))   \\
% (ii) & (\mathtt{rule}\ \Rightarrow  & (\mathbf{and}\ \phi(Trans)) (Top\_Rel\ \In\ \Ou\ State\ State'))  \\
  (iii)\ (\mathtt{rule}\ (\Rightarrow (I_{N_0}\ \In_{N_0}\ \Ou_{N_0}\ State_{N_0})\ (Main\ \In_{N_0}\ \Ou_{N_0}\ State_{N_0})) \\
  (iv)\ (\mathtt{rule}\ (\Rightarrow (\mathbf{and}\ (T_{N_0} \In'_{N_0}\
  \Ou'_{N_0}\ State_{N_0}\ State_{N_0}')\ (Main\ \In_{N_0}\ \Ou_{N_0}\ State_{N_0}))\ (Main\ \In'_{N_0}\ \Ou'_{N_0}\ State_{N_0}'))
\end{array}
$$

$(iii)$ characterizes the set of initial states while $(iv)$ defines the induction
step. Let $P_L$ be the expression representing the safety property for the
Lustre program $L=[N_0, N_1, \dots, N_m]$. As specified above, $P_L$ is a
predicate defined over the signature of the main node $N_0$. Let $P_H$ be its
equivalent in Horn clauses format, i.e. $P_H = \phi(P_L)$.  Then, we encode the
checking of the property $P_H$ on the Horn encoding as defined above in the
following manner:

\[
\begin{array}{lll}
% there is no rule for Error, it is a query
%(iii) & (\mathtt{rule}\ \Rightarrow\ &  \mathbf{False}\ Error) \\
(v) & (\mathtt{rule}\ \Rightarrow & (\mathbf{and}\ (Main\ \In_{N_0}\ \Ou_{N_0}\ State_{N_0})\
(\mathbf{not}\ P_H))\ Error)
\end{array}
\]

where $Error$ is the predicate marking the error state. Such state is reachable
if and only if the property is not valid. If its unreachable then the property is valid. Tools like Z3~\cite{pdr} are able to give a certificate for (un)reachability. A certificate of reachability is in a nutshell a proof of unsatisfiability. In this case the certificate is presented as a trace. A certificate for un-reachability is in a nutshell a model for the recursive predicates. The following theorem establishes a
correspondence relation between Lustre program and the compiled Horn clauses $H$
using the function \texttt{lus2horn}.

\begin{thm} The semantics of the Lustre program $L=[N_0,\dots,N_m]$ and the
  semantics of its Horn clauses encoding $H=\mathtt{lus2horn}(L)$ are in strong bisimulation.
\end{thm}

\begin{proof}A classical strong bisimulation proof only sketched here: the two set of initial
  states coincide while each transition that could be performed on one side is
  computable on the other: (i) for every execution of $L$ there is a derivation of $H$, and (ii) for every derivation of $H$ there is an equivalent execution of $L$.

\end{proof}
% \begin{itemize}
%   \item encode property. i.e., define predicate Main, or M, such that Init\_Top -> M, M /\ TopBody -> M
%    \item M (s) is derivable from above iff state 's' is reachable from init in lustre program
%     \item for verification, add clause M(error) -> false then, false is derivable iff error is reachable.
%   \end{itemize}

\subsection{Example}\label{ssec:example}
As an example of the compilation process \texttt{lus2horn}, we will consider a simple
Lustre program that compares two implementations of a 2-bit counter: a low-level
Boolean implementation and a higher-level implementation using integers. The
left hand side of Figure~\ref{fig:two} illustrate the Lustre code.  The
\texttt{greycounter} node internally repeats the sequence $ab = \{00, 01, 11,
10,00, . . .\}$ indefinitely, while the \texttt{integercounter} node repeats the
sequence $time = \{0, 1, 2, 3, 0, . . .\}$. In both cases the counter returns a
boolean value that is true if and only if the counter is in its third step or input variable $reset$ is true. The top node test is an example of a synchronous
observer. So we wish to verify the safety property that both implementations
have the same observable behavior, i.e. that the stream $OK$ is always true. On
the right hand side of Figure~\ref{fig:two} is the corresponding Horn clauses
encoding~\footnote{The variable \textit{time} on the Lustre program is denoted by the variable \textit{t} in the Horn encoded.}. The predicates $IC, GC$ and $T$ encode the transition relation of the
nodes \texttt{intcounter}, \texttt{greycounter} and \texttt{top}
respectively. While the predicate $IC\_Init, GC\_Init$ and $T\_Init$ encode the initial states of the three nodes. The predicate $M$ encode the main entry point of the two counters
Lustre program; while $Error$ is a predicate used to mark the error states.

\begin{figure}[!h]
\begin{minipage}{\linewidth}
\centering
\begin{minipage}{0.4\linewidth}
\begin{lustre}
node greycounter (reset:bool)
    returns (out:bool);
var a,b:bool;
let
 a = false $\rightarrow$
       (not reset and not pre(b));
 b = false $\rightarrow$ (not reset and pre(a));
 out = a and b;
tel

node intcounter (reset:bool)
     returns (out:bool);
var time: int;
let
 time = 0 $\rightarrow$
      if reset or pre(time) = 3
                 then 0
            else pre time + 1;
  out = (time = 2);
tel

node top (reset:bool)
     returns (OK:bool);
var b,d:bool;
let
  b = greycounter(reset);
  d = intcounter(reset);
  OK = b = d;
  --!PROPERTY : OK=true;
tel
\end{lustre}
\end{minipage}
\hspace{0.01\linewidth}
      \begin{minipage}{0.57\linewidth}
\begin{lustre}
(declare-rel GC(Bool Bool Bool Bool Bool Bool))
(declare-rel IC(Bool Bool Int Int))
; ... predicate declartions ...
(declare-rel Error())
; ... variable declarations ... ;

(rule ($\Rightarrow$ (and (= a' (and (not reset) (not b)))
               (= b' (and (not reset) a))
               (= out (and b' a')))
      (GC reset out a b a' b')))
(rule ($\Rightarrow$ (and (= a' false)
               (= b' false)
               (= out (and b' a'))
      (GC_Init reset out a' b')))
(rule ($\Rightarrow$ (and (= t'
            (ite (or reset (= t 3)) 0 (+ t 1)))
            (= out (= t' 2)))
      (IC reset out t t')))
(rule ($\Rightarrow$ (and (= t' 0) (= out (= t' 2)))
      (IC_Init reset out t')))
(rule ($\Rightarrow$ (and (GC reset gout ga gb ga' gb')
               (IC reset iout it it')
              (= ok (= iout gout)))
       (T reset ok ga gb it ga' gb' it')))
(rule ($\Rightarrow$ (and (GC_init reset gout ga' gb')
               (IC_init reset iout it')
              (= ok (= iout gout)))
      (T_Init reset ok ga' gb' it')))

(rule ($\Rightarrow$ (T_Init reset ok ga gb it)
    (M reset ok ga gb it)))

(rule ($\Rightarrow$ (and (M reset' ok' ga gb it)
           (T reset tok ga gb it ga' gb' it'))
        (M reset tok ga' gb' it')))
(rule ($\Rightarrow$ (and (M reset ok ga gb it)
              (not (= ok true)))
        Error)
(query Error)
\end{lustre}
% \ploc{For init, what if the node is a = not reset -> (not reset and not pre(b))?
% How could you definie init with the first input?}
\end{minipage}
\end{minipage}
% \ploc{Teme: you cannot have all outputs on M. Or you need to change the sig of
%   top. I changed it to M reset ok ga gb it. I also defined the ``init''
  % according to the rest of the section}
\caption{Two counters example.}
\label{fig:two}
\end{figure}

Using the PDR engine of Z3~\cite{pdr} we are able to get the modular invariant
of the Horn encoding for the predicate $IC$, $GC$, $M$ and $T$. For example, for the predicates $IC\_Init, IC$ and $GC\_Init, GC$ we obtain the following invariants:
\[
\begin{array}{lccl}
IC\_Init(reset, out, time) & = & & (time=0) \wedge \neg(out) \\
IC(reset, out, time, time') & = & & (time < 3 \to time' \geq 0) \\
& & \wedge &(out \to time \leq 1 ) \\
& & \wedge & (time' \leq 0 \vee \neg(time \leq 3) \vee \neg(time \geq 3)) \\
& & & (\neg(time \geq 3) \vee time'\geq 0) \\
\hline
GC\_Init(reset, out, a, b) & = && \neg(out) \wedge \neg(a) \wedge \neg(b)\\
GC(reset, out, a, b, a', b') & = & \wedge & b' \leftrightarrow a \\
& &  \wedge & \neg a' \leftrightarrow b \\
& & \wedge & \neg out \leftrightarrow \neg a \vee a'
\end{array}
\]

\noindent
For example, for the node \texttt{intcounter}, denoted by the predicates $IC\_Init$ and $IC$, we obtain that the variable \textit{time} is bound in the interval [$0,3$].
%%% Local Variables:
%%% mode: latex
%%% TeX-master: "pap"
%%% End:

%  LocalWords:  Lustre tuple Init lus stateful pre lcl pb uid Eq Expr Eqs cst
%  LocalWords:  expr NewVars InitVars Mem callee SMT eq eqs invariants lll ie
%  LocalWords:  un reachability unsatisfiability bisimulation computable GC IC
%  LocalWords:  greycounter bool intcounter ite ga gb iout ok init tok PDR lccl

\section{From monolithic to modular invariants}\label{sec:mono}
%!TEX root = pap.tex
In the last section we illustrated how we compile a Lustre program in a modular
fashion and obtain a Horn clause representation. Such encoding allows to
generate modular invariants by exploiting tools based on property-directed
reachability such as Z3~\cite{pdr}. In this section, we describe a technique to
synthesize a modular invariants given a monolithic invariant. The latter is an invariant over an inlined version of the program, in which all the nodes (non-recurisve predicate) are inlined to the main Lustre node. Formally, it is defined as follows:

\begin{defi}[Monolithic invariant]
Let $L=[N_0,\dots,N_m]$ be a Lustre program and $H=\mathtt{lus2horn}(L)$ be the
set of Horn clauses defined in the previous section. Let $K = \mathtt{inline(H)}$ an
inlined version of $H$. That is all the non-recursive predicate are inlined by
resolution. The function $\mathtt{inline}(H)$, will generate a tuple $M=(I_{N_0}, T_{N_0})$, where $I_{N_0}$ and $T_{N_0}$ are the predicates for the initial states and the transition relation as defined in $(i)$ and $(ii)$ of Section~\ref{sec:modular}. Let $P_K$ be a safety property. Checking the property $P_K$ over $M$ as defined in $(v)$ of Section~\ref{sec:modular}, we obtain an invariant $\pi : M \rightarrow \mathcal{F}$, where $\mathcal{F}$ is a set of formulas valid in $M$. We call $\pi$ a monolithic invariant.
\end{defi} 

Given a monolithic invariant $\pi$ defined as above, we are interested in obtaining modular invariants. That is, given a modular encoding of the program we would like to reconstruct the modular invariant from $\pi$. The following theorem states that given a monolithic invariant we can obtain a modular invariant for a modularly defined Horn clauses.

\begin{thm}[]
Let $\pi : M\rightarrow \mathcal{F}$ be a monolithic invariant for $M=(I_{N_0}, T_{N_0})$ of an inlined Horn clause $K = \mathtt{inline}(H)$. Then, $\pi$ can be extended to a model $\pi'$ for the Horn clause $H$; where $H$ is a modular set of Horn clauses as defined in Section~\ref{sec:modular}.
\end{thm}

Given a modular Horn clauses $H$ as defined by the rules $(i)$ and $(ii)$ in Section~\ref{sec:modular}, where $I_{N_0}$ and $T_{N_0}$ are the predicates for the initial and transition relation of the top node, we encode the following Horn rules in order to get a modular invariants:

\[
\begin{array}{crl}
(vi) & (\mathtt{rule}\ (\Rightarrow  \quad\quad\quad &  (I_{N_0} \In_{N_0}\ \Ou_{N_0}\ State_{N_0})\\
   & & (Mono\ \In_{N_0}\ \Ou_{N_0}\ State_{N_0})))\\
(vii) & (\mathtt{rule}\ (\Rightarrow  (\mathbf{and}\ &  (T_{N_0} \In'_{N_0}\
  \Ou'_{N_0}\ State_{N_0}\ State_{N_0}')\\
  & & (Mono\ \In_{N_0}\ \Ou_{N_0}\ State_{N_0})\\ 
  & & (\mathbf{not}\ (Mono\ \In'_{N_0}\ \Ou'_{N_0}\ State_{N_0}')))\\
  & Error))
\end{array}
\]

where $Mono$ is the predicate representing the monolithic invariant. Rule $(vi)$
encode the rule for the initial states, while rule $(vii)$ encode the reachability of the transition relation, where $Error$ is the predicate that marks the error state. By checking the reachability of the $Error$ state we can
obtain, as expected,  a certificate of un-reachability of it, producing a modular invariant for the predicates of in $H$.

\subsection{Example (cont.)}
Continuing the example from previous section, let \texttt{MONO} be the monolithic invariant for the inlined version of the \texttt{two\_counters}, defined as follows:

\begin{lustre}
(define-fun MONO ((reset Bool)(ok Bool) (ga Bool) (gb Bool) (it Int)) Bool
  (let ((tmp (not (or (not (<= it 0)) (not (>= it 0))))))
    (and ok
         (or tmp ga gb reset)
         (or (<= it 2) (not ga) (not reset))
         (<= it 3)
         (or (>= it 3) (not gb) ga)
         (or (>= it 2) (not gb) (not ga)))))
\end{lustre}

Given the modular definition of the Horn rules for $IC, GC$ and $T$ as defined in the sub-section~\ref{ssec:example}, we encode the reachability challenge in the following way:

\begin{lustre}
(rule (=> (and (T reset tok ga gb it ga' gb' it'))
               (MONO reset' ok ga gb it)
               (not (MONO reset' tok ga' gb' it')))
               Error))        
(query Error)
\end{lustre}

Z3 produces a certificate for such queries over the predicates $IC, GC$ and $T$.

% I(A, B, C) ==
%            And(Or(Not(Or(Not(B <= 0), Not(B >= 0))),
%                   Not(A <= 3),
%                   Not(A >= 3)),
%                Or(A >= 1, Not(C)),
%                Or(A <= 0, C, A >= 2),
%                Or(A <= 1, Not(C)),
%                Or(A >= 3, A + -1*B <= -1),
%                B + -1*A <= 1)),

% \[
% \begin{array}{lccl}
% IC\_Init(reset, out, t) & = & & (t=0) \wedge \neg(out) \\
% IC(reset, out, t, t') & = & & (t < 3 \to t' \geq 0) \\
% & & \wedge &(out \to t \leq 1 ) \\
% & & \wedge & (t' \leq 0 \vee \neg(t \leq 3) \vee \neg(t \geq 3)) \\
% & & & (\neg(t \geq 3) \vee t'\geq 0) \\
% \hline
% GC\_Init(reset, out, a, b) & = && \neg(out) \wedge \neg(a) \wedge \neg(b)\\
% GC(reset, out, a, b, a', b') & = & \wedge & b' \leftrightarrow a \\
% & &  \wedge & \neg a' \leftrightarrow b \\
% & & \wedge & \neg out \leftrightarrow \neg a \vee a'
% \end{array}
% \]

% LocalWords:  Lustre invariants reachability lus inlined modularly crl ok Bool
% LocalWords:  ga gb tmp IC GC tok

%%% Local Variables: 
%%% mode: latex
%%% TeX-master: "pap"
%%% End: 

\section{Minimization of modular interface}\label{sec:min}
%!TEX root = pap.tex

Given a Horn clause $H$ of the form $(\mathtt{rule}\ (\Rightarrow\ Body\ B))$,
where $Body$ is a conjunction of expressions and $B$ is a predicate, we can get an invariant $\pi
: B \mapsto \mathcal{F}$, which is of the form:

$$
B(\In, \Ou, State, State') = \bigwedge f 
$$

where $f \in \mathcal{F}$, and $\mathcal{F}$ is a set of formulas which makes
the predicate $B$ valid. Such invariant could be obtained using techniques from
property-based reachability~\cite{pdr} or other techniques for invariant
generators, e.g., template-based~\cite{KahGT-NFM-11} or abstract
interpretation-based~\cite{GarocheKT13}. In a nutshell, $\pi$ represents a set
of formulae which prescribe the behavior of that particular component. In other
words, it represents an interface (or contract) of the component. We are
interested in obtaining a smaller interface (or contract). Specifically, we are
interested in minimizing the set of state variables for $\pi$. This means,
minimizing the \textit{bits} representing the state variables. For example, if
the variables in $State$ are represented by bit vectors, we are interested in
minimizing the number of bits. More generally, we define the notion of 
\textit{ranks of state variables} as follows:

\begin{defi}[Rank of state variables]\label{def:bits}
  Let $\mathcal{V}$ be a set of state variables. We define the number of bits as
  a function $\mathtt{rank}: \mathcal P (\mathcal{V}) \rightarrow \mathbb{N} x
  \mathbb{N}$, from the set of state variables $V$ to a pair $(\bar{I},\bar{B})$,
  where $\bar{I}$ is the number of integers and $\bar{B}$ is the number of bits representing it.

\[
\begin{array}{rcl}
\mathtt{rank}: \mathcal P (\mathcal V) & \rightarrow& \mathbb N \times \mathbb N\\
 \{ v\} & \mapsto & \left\{
 \begin{array}{ll}(0,1) & \textrm{when}\ type(v) = Bool\\
 (1,0) & \textrm{when}\ type(v) = Int\\
 (0,n) & \textrm{when}\ type(v) = BitVector^n
\end{array}\right.\\
V & \mapsto & \sum_{v \in V} \mathtt{rank}(v)
\end{array}
\]
\end{defi}

The following are some examples of ranks of state variables:  
$$
\begin{array}{cl}
(*) & 
\begin{array}{|l|l|}
\hline
\mbox{ State Variables } & \mbox{ Rank }\\
\hline
P = \{v:Int, w:Int, b:Bool\} & \mathtt{rank}(P)=(2,1)\\
\hline
Q = \{b_1:Bool, b_2:Bool\} & \mathtt{rank}(Q)=(0,2)\\
\hline
R = \{w:BitVector^{1000}\} & \mathtt{rank}(R)=(0,1000)\\
\hline
\end{array}
\end{array}
$$

Given an invariant we are interested in finding another invariant which is
smaller, c.f. less rank of state variables. Formally, we define smaller invariant
using lexicographic order as follows:

\begin{defi}[Lexicographic order]\label{def:smaller}
  Let $\mathcal{V}$ be a set of state variables. Let $\mathtt{rank}: \mathcal{V}
  \rightarrow \mathbb{N} x \mathbb{N}$ be the function defined in
  Def.~\ref{def:bits}.  Let $State_H, State_K \subseteq\mathcal{V}$ be two sets
  of state variables. There is a lexicographic order $<_{State}$ over the set
  $\mathcal{V}$ such that
\[
State_K <_{State} State_H \triangleq 
\mathtt{rank}(State_K) <_{\mathbb N^2} \mathtt{rank}(State_H)
\]
where $\leq_{\mathbb N^2}$ denotes the usual lexicographic orderings of pairs of the
classical order: $(a,b) \leq_{\mathbb N^2} (c,d)\ \textrm{iff}\\ a
< c \lor (a = b \land b < d)$. The definition
is lifted to invariants in the following way:
\[
K(\In, \Ou, State_K, State_K') <_{State} H(\In, \Ou, State_H, State_H')
\textrm{ iff }
State_K <_{State} State_H.
\]
\end{defi}
For the predicates in example $(*)$ we have the ordering $Q <_{State} R <_{State} P$.
%
% for the previous example of predicates following state variables:
% $$
% \begin{array}{|l|l|}
% \hline
% \mbox{ State variables } &  \mbox{ Rank }\\
% \hline
% P =  \{x:Int, b:Bool\} & \mathtt{rank}(P)=(1,1)\\
% \hline
% Q =  \{b_1:Bool, b_2:Bool\} & \mathtt{rank}(Q)=(0,2)\\
% \hline
% R =  \{y:BitVector^{1000}\} & \mathtt{rank}(R)=(0,1000)\\
% \hline
% \end{array}
% $$
We now sketch our idea of minimization of invariants.
Let $H$ be the set of Horn clauses as defined in~Section~\ref{sec:modular}, for which we have obtained a set of modular invariants $\pi_0,\dots,\pi_m$ for the Horn rules $H_0, \dots, H_m$ of $H$.  Our aim is to generate a set of invariants $\pi'_0,\dots,\pi'_n$ such that the predicates of the Horn rules $H_0, \dots H_m$ have been minimized. That is, we generate predicates $Q_0,\dots,Q_m$ such that $Q_0 <_{State} P_0, \dots, Q_m <_{State} P_m$ where $P_0,\dots, P_m$ are the predicates for the Horn rules $H_0,\dots,H_m$.

In a nutshell, our approach of minimization is an iterative one which is based on a CEGAR~\cite{cegar} type technique. We start by abstracting the Horn rules by removing state variables from the signature of their predicates. Let $H$ be the set of Horn rules and $H^{\sharp}$ the abstracted one.  We check whether the safety property $P$ is valid in  $H^{\sharp}$ (c.f. $(v)$ in Section~\ref{sec:modular}). If the result is unsatisfiable then we are done, i.e., we have found certificate of unsatisfiability for smaller predicates, hence smaller invariants. If the result is satisfiable, it means we have a \textit{spurious counterexamples}, traces for $H^{\sharp}$ that falsify the property $P$ but are not legal traces of $H$. In this case we need to refine $H^{\sharp}$. 

Let $t^{\sharp}$ be a trace of $H^{\sharp}$ that falsifies the property $P$. However, the corresponding trace $t$ of $H$ is not legal, hence is unsatisfiable. This means that there is a constraint $Z$ in $t \backslash t^{\sharp}$. Therefore, we refine $H^{\sharp}$ by finding the smallest subset of $Z$ that makes $t^{\sharp}$ unsatisfiable.

\section{Conclusion}\label{sec:concl}
In this paper, we have described two techniques to synthesis modular
invariants for synchronous code encoded as Horn clauses. Modular invariants are
very useful for different aspects of analysis, synthesis, testing and program
transformation. We have described two techniques to generate modular invariants
for code written in the synchronous dataflow language Lustre. The first
technique directly encodes the synchronous code in a modular fashion. While in
the second technique, we synthesize modular invariants starting from a monolithic
invariant. Both techniques take advantage of analysis techniques based on
property-directed reachability. Both techniques have been implemented in a tool
that targets the verification of safety properties specified in Lustre. We have
also sketched a technique for minimizing invariants following a CEGAR-like
approach. In the future, we plan to fully work the details of the minimization
process and implement such techniques.  
%\nocite{*}
\bibliographystyle{eptcs}
\bibliography{biblio}

\end{document}